\RequirePackage[l2tabu, orthodox]{nag}
\RequirePackage{snapshot}

\documentclass[9pt,onecolumn]{extarticle}
\makeatletter
\if@twocolumn
  \usepackage[dvips,letterpaper,top=0.5in, bottom=0.5in, left=0.75in, right=0.5in,includefoot,heightrounded]{geometry}
\else
  \usepackage[dvips,letterpaper,margin=1in,includefoot,heightrounded]{geometry}
\fi

\usepackage{srcltx}

\usepackage{amsmath}
\usepackage{amssymb,amsfonts}

\usepackage{abstract}

\usepackage{graphicx}
\usepackage[usenames,dvipsnames]{color}
\usepackage[usenames,dvipsnames]{xcolor}
\usepackage{subfigure}

\usepackage{booktabs}

\usepackage{setspace}
\usepackage{flushend}

\usepackage{multicol}

\usepackage{mwe}

\usepackage{url}
\usepackage{bm}

\usepackage{enumerate}

\usepackage{hyperref}
\usepackage[hyperpageref]{backref}
\usepackage{blindtext}

\usepackage{fontawesome}
\usepackage{amsrefs}

\newtheorem{proposition}{Proposition}
\newtheorem{definition}{Definition}

\makeatletter

\newcommand{\printtitle}{%
\makeatletter
\if@twocolumn

\twocolumn[%
  \maketitle
  \begin{onecolabstract}
    \myabstract
  \end{onecolabstract}
  \begin{center}
    \small
    \textbf{Keywords}
    \\\medskip
    \mykeywords
  \end{center}
  \bigskip
]
\saythanks
\else
  \maketitle
  \begin{onecolabstract}
    \myabstract
  \end{onecolabstract}
  \begin{center}
    \small
    \textbf{Keywords}
    \\\medskip
    \mykeywords
  \end{center}
  \bigskip
  \onehalfspacing
\fi
\makeatother
}

\title{%
Multidimensional Wavelets for Scalable Image Decomposition:
Orbital Wavelets}

\author{%
H. M. de~Oliveira%
\thanks{%
Departamento de Estat\'{\i}stica,
Universidade Federal de Pernambuco, Recife, PE, Brazil.
E-mail: hmo@de.ufpe.br}
\and
V. V. Vermehren%
\thanks{%
Departamento de Engenharia El\'etrica,
Universidade do Estado do Amazonas, Manaus, AM,  Brazil}
\and
R. J. Cintra%
\thanks{%
Departamento de Estat\'{\i}stica,
Universidade Federal de Pernambuco, Recife, PE, Brazil.}
}

\date{}

\newcommand{\myabstract}{%
Wavelets are closely related to the Schr\"odinger's wave functions
and the interpretation of Born. Similarly to the appearance of atomic
orbital, it is proposed to combine anti-symmetric wavelets into orbital
wavelets.
The proposed approach allows the increase of the dimension of
wavelets through this process.
New orbital 2D-wavelets are introduced
for the decomposition of still images, showing that it is possible
to perform an analysis simultaneous in two distinct scales. An example
of such an image analysis is shown.
}

\newcommand{\mykeywords}{%
2D wavelets;
anti-symmetric wavelets;
orthogonal wavelets;
image analysis.
}

\begin{document}

\printtitle

\section{Introduction}
\label{sec:introduction}

Wavelet transform methods are important tools in image processing due to their capabilities for
multiresolution analysis and image decomposition~\cite{Mallat99,Burrus98}.
Wavelet-based image processing
find application
in computer graphics~\cite{Schroder96}, including radiosity, global
illumination~\cite{Stollnitz96}, and real volume data~\cite{Muraki93},
and volume rendering~\cite{Roerdink99};
being
routinely
considered as an approach
for texture image decomposition,
image coding,
subband coding,
fast image segmentation~\cite{Aujol2006,Antonini92,Vetterli90,Lim90,Kim2003},
and
3D signal processing~\cite{Taubman94}.
In particular,
wavelet transform coding~\cite{Ohm2004,Bottreau2001,Usevitch2001,DeVore95} has emerged as a
practical and fully-established tool~\cite{Mallat99,Aujol2006}
which
benefits image compressing methods,
such as
the {JPEG} 2000 standard~\cite{Skodras2001,Usevitch2001,Saha2000,Hilton94},
and
multimedia schemes on Internet,
such as
data streaming over {IP}~\cite{Radha2001,Baganne2003}
and
image querying~\cite{Jacobs95}.
Scalable coding for image, audio, and video
largely
adopts wavelet-based methods
as demonstrated in the MPEG-4 codec~\cite{Radha2001,Ohm2004,Bottreau2001,Shaar2003,Sodagar2002}.

Multiwavelets
have been explored
for the assessment of order/disorder
in reconstructed biomedical images~\cite{Zemni2019}.
Further connections between information-theoretical entropy
and
wavelet analysis
were explored in the context of geoscience~\cite{nicolis20152d}.
Recent advances in the field
of image decomposition
include
the proposition
of special filters for
spherical harmonics modeling~\cite{Jallouli2019b}
capable
of
multi-level decomposition
suitable for 3D images
with the introduction
of the concept of spherical harmonics entropy~\cite{Jallouli2019a}.

In the same vein of exploring connections between different research fields,
Ashmead reported~\cite{Ash2012}
a link between
quantum mechanics and wavelets.
Despite the potentially deep mathematical meaning of such link,
it has not been significantly explored in
the
context of wavelet decomposition and image analysis.
Indeed,
the wave nature of light can be deduced from the phenomenon of interference, the photoelectric effect, however, it seems to suggest a corpuscular nature of light.
Theoretical physicists struggled to include observations like the photoelectric effect and the wave-particle duality into their formulations~\cite{Sully2012}.
Erwin Schr\"odinger
employed advanced mechanics to address
such phenomena
and developed an equation that relates the space-time in quantum mechanics. %
Because wavelets are localized in both time and frequency they offer significant advantages for the
analysis of problems in quantum mechanics.

In this paper,
we aim at proposing
an alternative wavelet decomposition scheme.
For such,
we explore
the above discussed link between
wavelets and quantum mechanics
and
shed some light on some of these relations.
Rather than seeking at wavelet features on particles or waves,
we adapted some concepts of quantum mechanics to a novel wavelet decomposition for still images.

The paper is organized as follows.
Section~\ref{sec:motivation}
details the main ideas behind the proposed wavelet system
which
is based on
an interpretation
from particle physics and quantum states~\cite{Eisberg2007,Beiser94}.
We describe the construction of the proposed orbital wavelets
for the image decomposition.
It is shown that the introduced derivation
is naturally
suitable for generating two-dimensional wavelets.
In Section~\ref{sec:implementing}
we submit standard imagery to the proposed wavelet decomposition
scheme
and compare
the results with
standard wavelet decomposition..
Section~\ref{sec:concluding}
concludes the paper.

\section{Orbital Wavelet Decomposition}
\label{sec:motivation}

\subsection{Particle Systems}

The wave functions describing electronic orbitals
can be combined generating
atomic orbitals.
Let us consider
a
two-particle non-interaction systems
with
particles
located at position vectors
$\mathbf{r}_1$
and
$\mathbf{r}_2$,
respectively,
where
each particle
is equipped
with wave functions
$\psi_\alpha(\mathbf{r}_1)$
and
$\psi_\beta(\mathbf{r}_2)$
at states
$\alpha$ and $\beta$,
respectively.
The wave function
that characterizes the orbital interaction of the two particles
is furnished by a combination
$\psi_\alpha(\mathbf{r}_1)$
and
$\psi_\beta(\mathbf{r}_2)$
in two different configurations:
symmetric ($S$)
and
anti-symmetric ($A$)
combinations~\cite{Duck98}.
Such combinations are given by:
\begin{align}
\psi_{S}(\mathbf{r}_1,\mathbf{r}_2)
&
=
\frac{1}{\sqrt{2}}
\Big[
\psi_\alpha(\mathbf{r}_1)
\cdot
\psi_\beta(\mathbf{r}_2)
+
\psi_\alpha(\mathbf{r}_2)
\cdot
\psi_\beta(\mathbf{r}_1)
\Big]
,
\label{eq:4}
\\
\psi_{A}(\mathbf{r}_1,\mathbf{r}_2)
&
=
\frac{1}{\sqrt{2}}
\Big[
\psi_\alpha(\mathbf{r}_1)
\psi_\beta(\mathbf{r}_2)
-
\psi_\alpha(\mathbf{r}_2)
\psi_\beta(\mathbf{r}_1)
\Big]
\label{eq:5}
,
\end{align}
respectively.
The anti-symmetric case
can be conveniently
written
in matrix format
as
$
\psi_{A}(\mathbf{r}_1,\mathbf{r}_2)
=
\frac{1}{\sqrt2}
\operatorname{det}
\left[
\begin{smallmatrix}
\psi_{\alpha}(\mathbf{r}_1) & \psi_{\alpha}(\mathbf{r}_2)\\
\psi_{\beta}(\mathbf{r}_1) & \psi_{\beta}(\mathbf{r}_2)
\end{smallmatrix}
\right]
.
$

A
comparable concept
in the scope of wavelets, also characterized by wave functions,
would
be the combination of different spatial wavelets~\cite{DeO}.

\subsection{Orbital Wavelets}

\subsubsection{Symmetric Case: Standard Wavelets}

Usual wavelet image analysis combines one-dimensional (1D) wavelets to generate
a two-dimensional (2D) wavelet~\cite{Antonini92,Vetterli90}.
This can be done by considering a scaling function $\varphi$ and
a wavelet function $\psi$,
a version for each dimension, abscissa and ordinate~\cite{Burrus98}.
Thus,
we have a 2D scale function
$\varphi_\text{LL}(x,y) = \varphi(x)\cdot\varphi(y)$
and three
2D wavelet functions:
\begin{align}
\label{equation-standard-2d-wavelets}
\begin{split}
\psi_\text{LH}(x,y) &= \varphi(x)\cdot\psi(y)
,
\\
\psi_\text{HL}(x,y) &= \psi(x)\cdot\varphi(y)
,
\\
\psi_\text{HH}(x,y) &= \psi(x)\cdot\psi(y)
.
\end{split}
\end{align}

The wavelets
$\psi_\text{LH}(x,y)$
and
$\psi_\text{HL}(x,y)$
naturally exhibit
reflection symmetry
with respect
to the plane
$x=y$,
i.e.
$
\psi_\text{LH}(x,y)
=
\psi_\text{LH}(y,x)
$
and
$
\psi_\text{HL}(x,y)
=
\psi_\text{HL}(y,x)
$.
Such reflection symmetry
is analogous
to
the symmetric wave function
described in~\eqref{eq:4}.

\subsubsection{Anti-symmetric Case: Orbital Wavelets}

Following this analogy,
considering a single orthogonal wavelet system,
we are compelled
to pursue the definition of wavelets
that
could be regarded as the counterparts
of the
anti-symmetric wave function described in~\eqref{eq:5}.
Thus
the combination of $\varphi(\cdot)$ and $\psi(\cdot)$
should
be arranged to provide anti-symmetry,
i.e.
the
$\hat{\psi}_\text{LH}(x,y)$
and
$\hat{\psi}_\text{HL}(x,y)$
wavelets
should be such that:
\begin{align}
\label{eq:2}
\hat{\psi}_\text{LH}(x,y)
&
=
-
\hat{\psi}_\text{LH}(x,y)
\\
\hat{\psi}_\text{HL}(x,y)
&
=
-
\hat{\psi}_\text{HL}(x,y)
\end{align}
A solution to the above requirement
is
to define an orbital-inspired combination
(cf.~\eqref{eq:5})
of $\varphi$ and $\psi$
according to the following:
\begin{align}
\label{equation-orbital}
\hat{\psi}_\text{LH}(x,y)
\triangleq
\frac{1}{\sqrt{2}}
\Big[
\varphi(x)\cdot\psi(y)-\psi(x)\cdot\varphi(y)
\Big]
.
\end{align}
In order
to illustrate the effect
of this definition
we consider
the case of the Meyer orthogonal wavelet~\cite{Meyer90}.
Figure~\ref{fig:meyer} shows
the surface plots for
the standard  of the discussed functions.

\begin{figure}
\centering

\subfigure[$\psi_\text{LH}(x,y)$]
{\includegraphics[width=0.48\columnwidth]{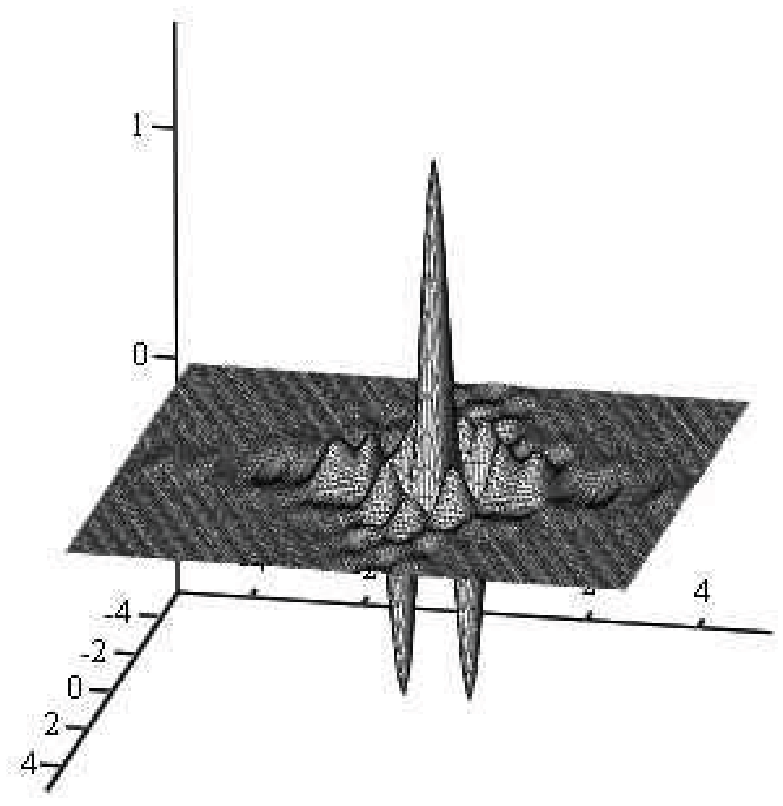}}
\subfigure[$\hat{\psi}_\text{LH}(x,y)$]
{\includegraphics[width=0.48\columnwidth]{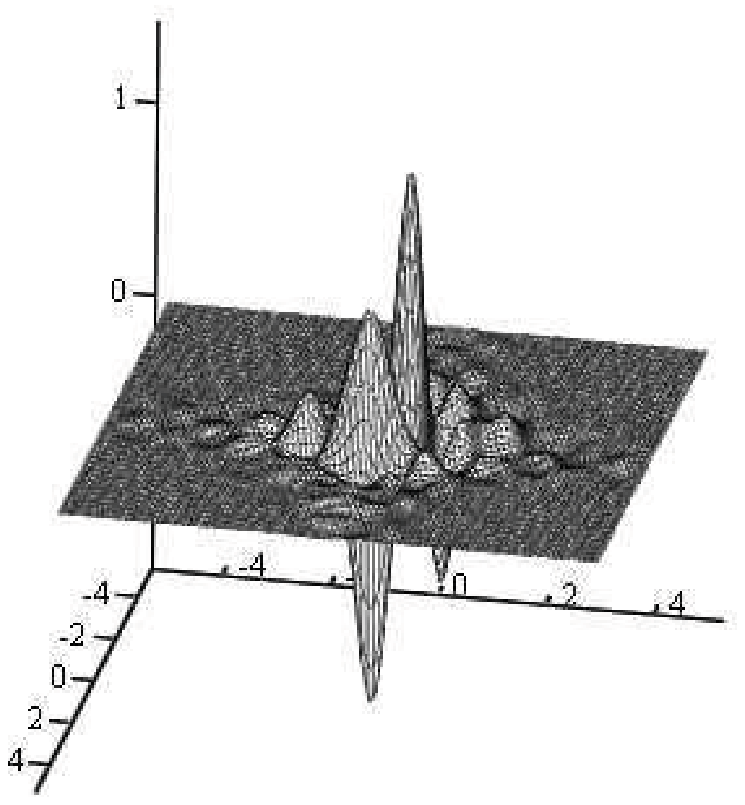}}

\vspace*{8pt}

\caption{2D-Meyer decomposing functions of the secondary diagonal:
(a)~$\psi_\text{LH}(x,y)$;
(b)~$\hat{\psi}_\text{LH}(x,y)$.}

\label{fig:meyer}
\end{figure}

Analogously,
we define
the
wavelets
related to the LL and HH
decompositions
according
to the following
expressions:
\begin{align}
\label{equation-orbital-wavelet-LL}
\hat{\varphi}_\text{LL}(x,y)
&
\triangleq
\frac{1}{\sqrt{2}}
\Big[
\varphi^{\ast}(x)
\cdot
\varphi(y)
+
\varphi^{\ast}(y)
\cdot
\varphi(x)
\Big]
,
\\
\label{equation-orbital-wavelet-HH}
\hat{\psi}_\text{HH}(x,y)
&
\triangleq
\frac{1}{\sqrt{2}}
\Big[
\psi^{\ast}(x)
\cdot
\psi(y)
-
\psi^{\ast}(y)
\cdot
\psi(x)
\Big]
.
\end{align}
The above definitions allows
the analysis of images using continuous
complex wavelets~\cite{Kingsbury99}.
For real-valued
wavelets,
the above expressions
collapse to the usual forms
$\varphi(x) \cdot \varphi(y)$
or
$\psi(x)\cdot\psi(y)$
present in standard wavelet analysis.
Thus in the real case,
although
the proposed wavelet
$\hat{\varphi}_\text{LL}(x,y)$
in \eqref{equation-orbital-wavelet-LL}
is well-defined,
the
wavelet
$\hat{\psi}_\text{HH}(x,y)$
would collapse to zero.

An approach to address such degeneracy
is to consider
daughter wavelets at different
scales.
Therefore
let us consider the 1{D} orthogonal~\cite{Lawton91,
Maab96} wavelet mother $\psi(x)$ equipped
equipped
with her daughter wavelets
$\{\psi_{a,b}(x)\}_{a\neq0,b\in R}$.
The formalism shown in~\eqref{equation-orbital-wavelet-HH}
can be extended
by considering
the inclusion
of
two wavelets $\psi_{a_1,b}(\cdot)$ and $\psi_{a_2,b}(\cdot)$
resulting in the following definition.

\begin{definition}
\label{definition-2d-orbital}
The function \mbox{2D}-orbital at the scales $\{a_1,a_2\}$
is defined by:
\begin{align}
\hat{\psi}_\text{HH}(x,y)\triangleq\frac{1}{\sqrt{2}}\operatorname{det}\begin{bmatrix}\psi_{a_1,b}^{*}(x) & \psi_{a_1,b}^{*}(y)\\
\psi_{a_2,b}(x) & \psi_{a_2,b}(y)
\end{bmatrix},
\end{align}

which can be rewritten as:
\begin{align}
\hat{\psi}_\text{HH}(x,y)
=
&
\frac{1}{\sqrt{2|a_1||a_2|}}
\psi^{\ast}\left(\frac{x-b}{a_1}\right)
\cdot
\psi\left(\frac{y-b}{a_2}\right)
\\
&
-
\frac{1}{\sqrt{2|a_1||a_2|}}
\psi\left(\frac{x-b}{a_2}\right)
\cdot
\psi^{\ast}\left(\frac{y-b}{a_1}\right)
.
\end{align}
\end{definition}

The condition $a_1\neq a_2$ eliminates the degeneration
$\hat{\psi}_\text{HH}(x,y)=0$.
This is to some extent in connection to the
Pauli Exclusion Principle~\cite{Duck98},
which
states that with a single-valued many-particle wave function is equivalent to requiring the wave function to be antisymmetric.
An antisymmetric two-particle state is represented as a sum of states in which one particle is in state $\alpha$ and the other in state $\beta$. Besides, the relationship $\hat{\psi}_\text{HH}(y,x)=-\hat{\psi}_\text{HH}(x,y)$ ensures the desired asymmetry.
Here, we use the same wavelet-mother,
but at different scales.
The 2D decomposition stated in
Definition~\ref{definition-2d-orbital} results in
a strict 2D wavelet.

\subsection{Mathematical Properties}

\begin{proposition}
The previously defined 2D-orbital function has oscillatory behavior
satisfying the following properties:
\begin{enumerate}[(i)]
\item
$\intop_{-\infty}^{\infty}\hat{\psi}_\text{HH}(x,y)\operatorname{d}\!x=0$,
\item
$\intop_{-\infty}^{\infty}\hat{\psi}_\text{HH}(x,y)\operatorname{d}\!y=0$,
\item
$\intop_{-\infty}^{\infty}\intop_{-\infty}^{\infty}\hat{\psi}_\text{HH}(x,y)\operatorname{d}\!x\operatorname{d}\!y=0$.
\end{enumerate}

\end{proposition}

\begin{proof}It follows that

\begin{align}
\intop_{-\infty}^{\infty}\hat{\psi}_\text{HH}(x,y)\operatorname{d}\!x=\frac{1}{\sqrt{2}}\psi_{a_2,b}(y)\cdot\overline{\psi_{a_1,b}^{\ast}(x)}
-\frac{1}{\sqrt{2}}\overline{\psi_{a_2,b}(x)}\cdot\psi_{a_1,b}^{\ast}(y),
\end{align}
where
\begin{align}
\overline{\psi_{a,b}(x)}
\triangleq
\intop_{-\infty}^{\infty}
\psi_{a,b}(x)\operatorname{d}\!x
.
\end{align}
Therefore,
the property~(i) derives from the fact that $\psi_{a,b}(x)$,
$a=\{a_1,a_2\}$ be individual wavelets.
The demonstration for property~(ii)
is similar, considering that
\begin{align}
\intop_{-\infty}^{\infty}\hat{\psi}_\text{HH}(x,y)\operatorname{d}\!y
=
\frac{1}{\sqrt{2}}
\overline{\psi_{a_2,b}(y)}
\cdot
\psi_{a_1,b}^{\ast}(x)
-
\frac{1}{\sqrt{2}}
\psi_{a_2,b}(x)
\cdot
\overline{\psi_{a_1,b}^{\ast}(y)}.
\end{align}
The condition
$\intop_{-\infty}^{\infty}
\intop_{-\infty}^{\infty}
\hat{\psi}_\text{HH}(x,y)\operatorname{d}\!x\operatorname{d}\!y=0$
follows from Fubini's theorem~\cite{DeFig2000},
regardless of the order of integration.
\end{proof}

Hereafter
we assume
an orthogonal wavelet system
with unitary energy.
In other words,
the following conditions hold true:
\begin{enumerate}[(i)]

\item
the inner product
$\langle\psi_{a_1,b},\psi_{a_2,b}\rangle=0$,
i.e., the following integrals cancel out $\forall a_1\neq a_2$:
\begin{align}
\intop_{-\infty}^{\infty}
\psi_{a_1,b}(x)\cdot\psi_{a_2,b}^{*}(x)
\operatorname{d}\!x
=
\intop_{-\infty}^{\infty}
\psi_{a_1,b}^{*}(x)\cdot\psi_{a_2,b}(x)
\operatorname{d}\!x=0
,
\end{align}
and
\item
$\int_{-\infty}^\infty | \psi_{a_1,b}(x) |^2 \operatorname{d}\!x = 1$.

\end{enumerate}
It is also noteworthy that
\begin{align}
\langle\psi_{a_1,b},\psi_{a_2,b}\rangle^{*}
=
\langle\psi_{a_2,b},\psi_{a_1,b}\rangle
.
\end{align}

\begin{proposition}
The 2D-orbital functions have normalized energy.
\end{proposition}
\begin{proof}
Expanding the expression
$|\hat{\psi}_\text{HH}(x,y)|^2=\hat{\psi}_\text{HH}(x,y)\cdot\hat{\psi}_\text{HH}^{\ast}(x,y)$
yields
cross-product terms
in the following form
\begin{align}
\operatorname{cross}(x,y)
\triangleq
-\psi_{a_1,b}(x)
\cdot
\psi_{a_2,b}^{\ast}(y)
\cdot
\psi_{a_1,b}(y)
\cdot
\psi_{a_2,b}^{^{\ast}}(y)
\end{align}
and its complex conjugate
$\operatorname{cross}^\ast(x,y)$.
Performing the integration with respect
to $x$ and $y$
yields:
\begin{align}
\intop_{-\infty}^{\infty}
\operatorname{cross}(x,y)
\operatorname{d}\!x
=
-2
\cdot
\psi_{a_2,b}^{\ast}(y)\psi_{a_1,b}(y)
\cdot
\intop_{-\infty}^{\infty}
\psi_{a_1,b}^{\ast}(x)\psi_{a_2,b}(x)
\operatorname{d}\!x
,
\end{align}
and
\begin{align}
\intop_{-\infty}^{\infty}
\operatorname{cross}(x,y)
\operatorname{d}\!y
=
-2
\cdot
\psi_{a_2,b}^{\ast}(x)\psi_{a_1,b}(x)
\cdot
\intop_{-\infty}^{\infty}
\psi_{a_1,b}^{\ast}(y)\psi_{a_2,b}(y)\operatorname{d}\!x
.
\end{align}
Invoking the orthogonality condition
we obtain
that
all cross terms are void.
The remaining possibly nonnull terms
are:
\begin{align}
|\hat{\psi}_\text{HH}(x,y)|^2
=
\frac{1}{2}|\psi_{a_1,b}(x)|^2
\cdot
|\psi_{a_2,b}(y)|^2
\cdot
|\psi_{a_1,b}(y)|^2
\cdot
|\psi_{a_2,b}(x)|^2
\end{align}
and therefore,
because of the normalized energy condition,
we have:
\begin{align}
\intop_{-\infty}^{\infty}\intop_{-\infty}^{\infty}|\hat{\psi}_\text{HH}(x,y)|^2\operatorname{d}\!x\operatorname{d}\!y=1
.
\end{align}
\end{proof}

It is possible (more
easily) to combine orthogonal 1D-wavelets and use them to build a
new 2D-wavelet.

\begin{proposition}
The 2D-orbital function is a 2D wavelet.\end{proposition}

\begin{proof}
Let $\Psi(\omega)$
and
$\Psi_{a,b}(\omega)$ be the Fourier transforms
of
the wavelet~$\psi(t)$
and
the
daughter wavelets~$\psi_{a,b}(t)$,
respectively.
If
the admissibility condition holds~\cite{Boggess09,Burrus98},
\begin{align}
\intop_{-\infty}^{\infty}
\frac{|\Psi(\omega)|^2}{|\omega|}
\operatorname{d}\!\omega < \infty
\end{align}
then
$\intop_{-\infty}^{\infty}
\frac{|\Psi_{a,b}(\omega)|^2}{|\omega|}
\operatorname{d}\!\omega
<\infty
$,
since
$\Psi_{a,b}(\omega)=\sqrt{|a|}\,\Psi(a\omega)e^{-j \omega b}$~\cite{Mallat99}.
Let us now evaluate the condition for the 2D case. If the Fourier
transform pair
$\hat{\psi}_\text{HH}(x,y)\leftrightarrow\hat{\Psi}_\text{HH}(u,v)$ does exist,
then the 2D spectrum of $\hat{\psi}_\text{HH}~$ can be computed
in terms of the Fourier spectrum
of $\psi$:
\begin{align}
\hat{\Psi}_\text{HH}(u,v)
=
\frac{\sqrt{|a_1a_2|}}{\sqrt{2}}
\Big[
\Psi(a_1u)\Psi^{\ast}(a_2v)-\Psi(a_2u)\Psi^{\ast}(a_1v)
\Big]
.
\end{align}

From the generalized Parseval-Plancherel energy theorem~\cite{DeFig2000,
Burrus98}, the cross-terms vanish due to the orthogonality.
Thus, we have
\begin{align}
|\hat{\Psi}_\text{HH}(u,v)|^2
=
\frac{|a_1a_2|}{2}
|\Psi(a_1u)|^2
\cdot
|\Psi(a_2v)|^2
+
\frac{|a_1a_2|}{2}
|\Psi(a_2u)|^2
\cdot
|\Psi(a_1v)|^2
.
\end{align}
Then,
from the marginal admission conditions of the 1D daughter wavelets,
we
have that
\begin{align}
\int_{-\infty}^{\infty}
\int_{-\infty}^{\infty}
\frac{|\hat{\Psi}_\text{HH}(u,v)|^2}{|u|\cdot|v|}
\operatorname{d}\!u
\operatorname{d}\!v
<
\infty.
\end{align}
\end{proof}

\section{Orbital Wavelet-based  Image Decomposition}
\label{sec:implementing}

The standard 2D wavelet decomposition
effects coefficient matrices representing
vertical,
horizontal,
and diagonal structures
denoted by
sub-images
$\{
L_1L_1,
L_1H_1,
H_1L_1,
H_1H_1
\}$
and
$\{
L_2L_2,
L_2H_2,
H_2L_2,
H_2H_2
\}$
for
the first and second level decompositions~\cite{Strang1997}.
Figure~\ref{table-standard-decomposition}
illustrate the structures.
Noticed that the terms
$L_{1}L_{1}$,
$L_{2}L_{2}$,
$H_{1}H_{1}$,
and
$H_{2}H_{2}$ on the main diagonal correspond to a part of the standard
wavelet decomposition of the image into two levels.
The proposed wavelet analysis
results in a similar structure
with sub-images
$\{
\widehat{L_1L_1},
\widehat{L_1H_1},
\widehat{H_1L_1},
\widehat{H_1H_1}
\}$
and
$\{
\widehat{L_2L_2},
\widehat{L_2H_2},
\widehat{H_2L_2},
\widehat{H_2H_2}
\}$
in a two-level decomposition.
Despite the similarity,
the
sub-images are computed
according to the orbital wavelets.
The resulting scheme is depicted in Figure~\ref{table-orbital-decomposition}.

\begin{figure}
\centering

\subfigure[First level]{\includegraphics{std-wav-level-1.ps}}
\subfigure[Second level]{\includegraphics{std-wav-level-2.ps}}
\caption{Image decomposition scheme according to the standard \mbox{2D} wavelet decomposition.}
\label{table-standard-decomposition}
\end{figure}

\begin{figure}
\centering

\subfigure[First level]{\includegraphics{orb-wav-level-1.ps}}
\subfigure[Second level]{\includegraphics{orb-wav-level-2.ps}}
\caption{Image decomposition scheme according to the 2D orbital wavelet decomposition.
A two-level wavelet-orbital decomposition.}
\label{table-orbital-decomposition}
\end{figure}

Considered the symlet wavelet of order~4 (\texttt{symlet4})~\cite{Daub92},
we applied
the proposed decomposition
to the standard image \texttt{woman}~\cite{Misiti2002};
the resulting sub-images are shown
in Figure~\ref{fig:woman}.
For a qualitative comparison,
we included the sub-images obtained from the usual wavelet decomposition.
The computation was performed
in the Matlab environement~\cite{Misiti2002}.
It is worth nothing that the subtraction of the images
resulting from
$\psi_\text{LH}(x,y)$
and
$\psi_\text{HL}(x,y)$
results
in the image obtained by the wave function
in~\eqref{equation-orbital}.

\begin{figure}
\centering

\subfigure[standard wavelet decomposition]
{
\begin{tabular}{cc}
\includegraphics[width=0.2\columnwidth]{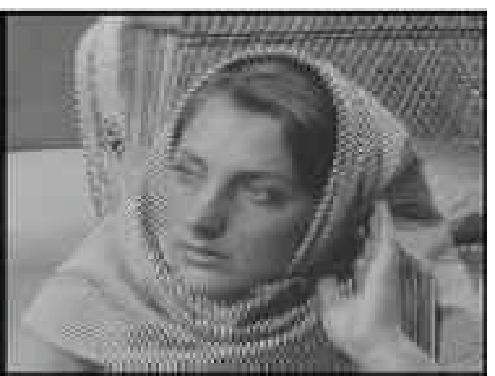}
&
\includegraphics[width=0.2\columnwidth]{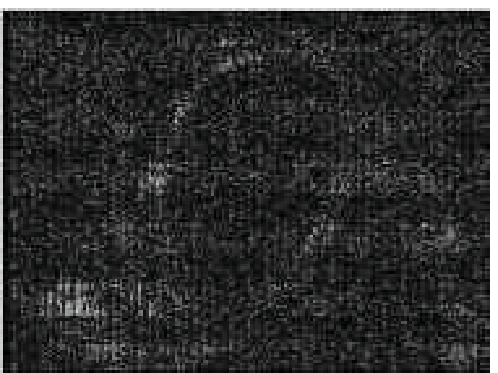}
\\
\includegraphics[width=0.2\columnwidth]{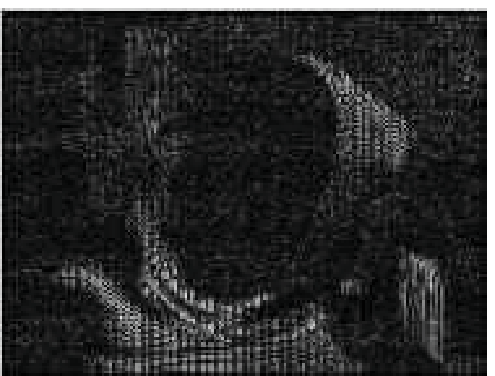}
&
\includegraphics[width=0.2\columnwidth]{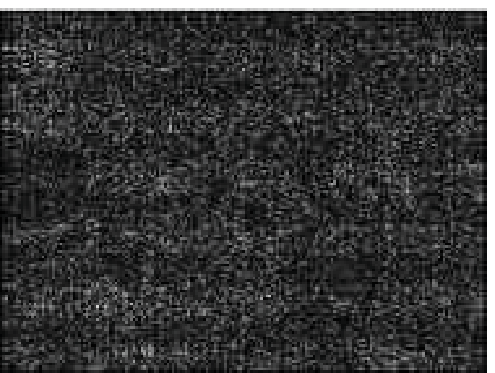}
\end{tabular}
}
\subfigure[Proposed decomposition]
{
\begin{tabular}{cc}
\includegraphics[width=0.2\columnwidth]{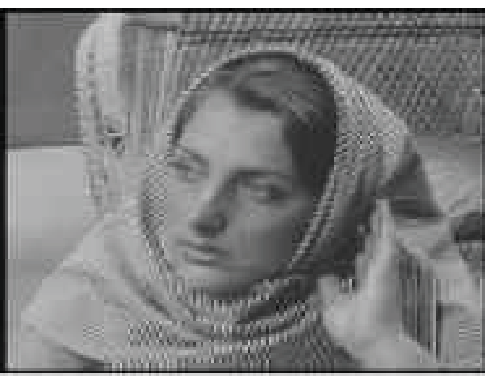}
&
\includegraphics[width=0.2\columnwidth]{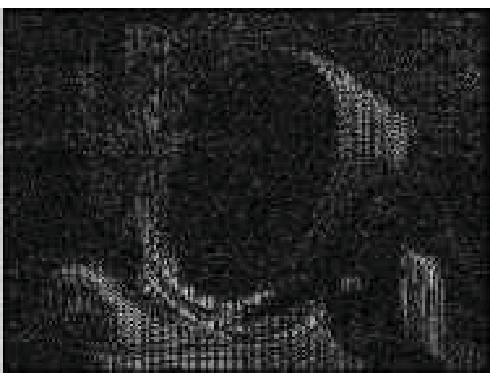}
\\
\includegraphics[width=0.2\columnwidth]{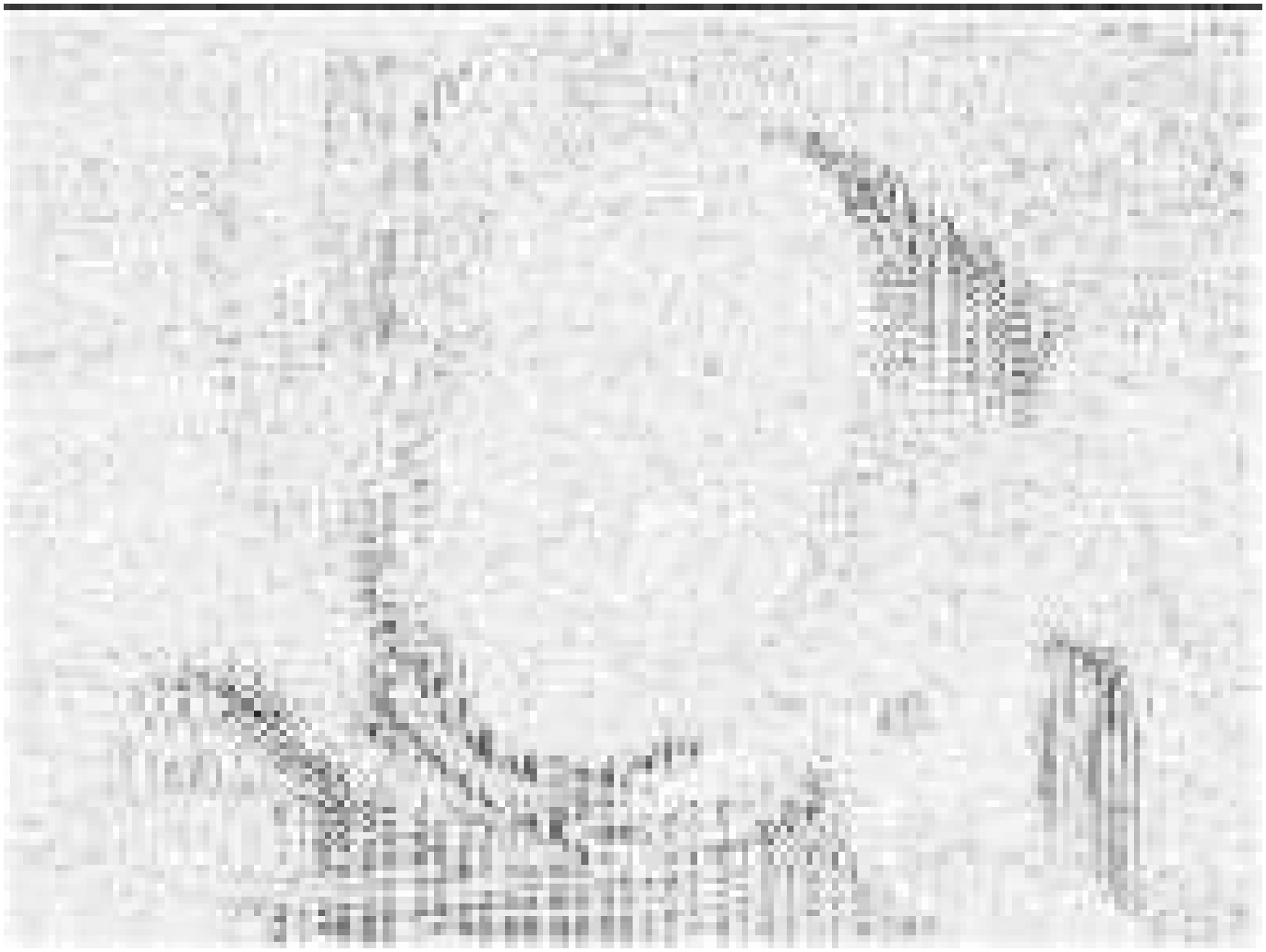}
&
\includegraphics[width=0.2\columnwidth]{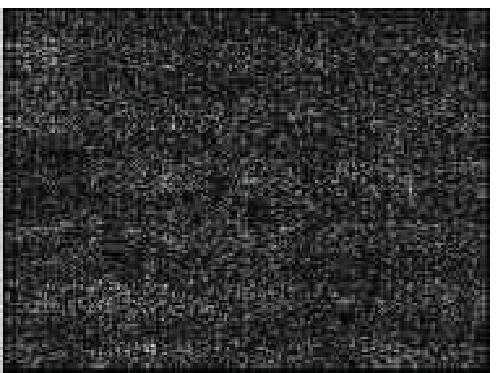}
\end{tabular}
}

\vspace*{8pt}

\caption{First-level decomposition woman image using \texttt{symlet4}
wavelet as defined in Matlab
according to
(a)~the standard wavelet decomposition
and
(b)~the proposed decomposition.}
\label{fig:woman}
\end{figure}

\section{Concluding Remarks and Future Work}
\label{sec:concluding}

This paper offers an alternative
approach for image
decomposition engendered by asymmetric orthogonal wavelets
whose definition
is inspired by the wave function theory
from particle physics.
Despite the focus being essentially on still image analysis,
the proposed
approach allows a fully scalable multimedia decomposition.
It remains
to be investigated the potential of this methodology
for
image compressing~\cite{DeVore95},
in 3D processing and scalable coding for multimedia
schemes~\cite{Ohm2004}.
Applications in other scenarios such as
wavelet-based watermarking~\cite{JHU2002}
or
steganography~\cite{Carrion2008}
also deserve an investigation.
To the best of our knowledge,
this is the first work to link
the exchanging formalism
from particle wave functions to wavelets analysis.

\section*{Acknowledgments}
The second and third authors
acknowledge
the partial support
from the Brazilian agencies
CAPES and CNPq,
respectively.

%\onecolumn
%

%

%
%
%
%
%
%

{\small
\singlespacing
\bibliographystyle{siam}
\bibliography{ref,ref-rev-1}
}

\end{document}